\documentclass[11pt]{article}

\usepackage{wrapfig}
\usepackage{graphicx}
\usepackage{palatino}
\usepackage{amsmath}               % great math stuff
\usepackage{amsfonts}              % for blackboard bold, etc
\usepackage{amsthm}                % better theorem environments
\usepackage{amssymb}
\usepackage{todonotes}
\usepackage[pagebackref]{hyperref}
\usepackage[margin=1in]{geometry}
\usepackage{todonotes}

\newtheorem{thm}{Theorem}[section]
\newtheorem{lem}[thm]{Lemma}

\renewcommand{\leq}{\leqslant}

\renewcommand{\geq}{\geqslant}

\DeclareMathOperator{\cala}{\mathcal{A}}

\begin{document}

\title{{\bf  APX-Hardness of Maximizing Nash Social Welfare with Indivisible Items}}

\author{
Euiwoong Lee\thanks{Supported by a Samsung Fellowship and Venkat Guruswami\rq{}s NSF CCF-1115525. {\tt euiwoonl@cs.cmu.edu} }}

\date{Computer Science Department \\ Carnegie Mellon University \\ Pittsburgh, PA 15213.}

\maketitle

\begin{abstract}
We study the problem of allocating a set of indivisible items to agents with additive utilities to maximize the Nash social welfare. 
Cole and Gkatzelis~\cite{CG15} recently proved that this problem admits a constant factor approximation.
We complement their result by showing that this problem is APX-hard.
\end{abstract}

\section {Introduction}
Suppose there are set $A$ of $n$ agents and set $B$ of $m$ items ($n \leq m$), and each agent $a \in A$ has a non-negative utility $u_{a,i}$ for each item $i \in B$. 
An allocation $\cala$ is defined to be a partition of $B$ into $n$ disjoint subsets $\{ B_a \}_{a \in A}$, 
and agent $i$\rq{}s utility under this allocation is $u_a(\cala) := \sum_{i \in B_a} u_{a, i}$. 
We study the problem of computing an allocation $\cala = \{ B_a \}_{a \in A}$ that maximizes the geometric mean of the agents\rq{} utilities $(\prod_{i \in [n]} u_i(\cala))^{1/n}$. 
The above objective function is known as the Nash social welfare defined in the fifties. 

There are two related objective functions. 
The first objective seeks to maximize the arithmetic mean, which is computationally very easy.
The second objective, also known as the Santa Clause problem~\cite{BCG09}, tries to compute an allocation that maximizes the minimum utility of any agent.
While the first objective focuses on efficiency and the second objective emphasizes fairness, the Nash social welfare balances these two extremes and satisfies other desirable properties. See Cole and Gkatzelis~\cite{CG15} for further discussions on the Nash social welfare. 

Very recently, Cole and Gkatzelis~\cite{CG15} suggested the first constant-factor approximation algorithm, which guarantees $2e^{1/e} \approx 2.889$-approximation. While only NP-hardness of computing the exact optimum is known~\cite{NNRR14}, we prove that there exists a constant $\mu > 1$ such that it is NP-hard to approximate this problem within factor $\mu$, suggesting that constant-factor approximation is the best we can hope for.

\section {APX-Hardness}
We reduce Vertex Cover on 3-regular graphs. Chleb\'{i}k and Chleb\'{i}kov\'{a}~\cite{CC06} proved that 
\begin{thm}~\cite{CC06}
Given a $3$-regular graph with $N$ vertices with $M = 1.5N$ edges, it is NP-hard to distinguish whether it has a vertex cover of size at most $c_{min}N$ for $c_{min} \approx 0.5103$, or every vertex cover has size at least $c_{max}N$ for $c_{max} \approx 0.5155$. Therefore, it is NP-hard to approximate Vertex Cover on 3-regular graphs within factor $\approx 1.01$. 
\end{thm}

Given a $3$-regular graph $G = (V, E)$, our reduction produces the following instance. Let $\alpha$ be a constant in $(\frac{1}{3}, \frac{1}{2})$. Any choice (say 0.4) is sufficient for our purpose, but we denote it as $\alpha$ to keep the presentation clear. 
\begin{itemize}
\item For each vertex $v$, there is an agent $a(v)$. Call them {\em vertex agents}. 
\item For each edge $e$, there is an agent $a(e)$. Call them {\em edge agents}.

\item There are $c_{min}N$ identical items. Each of them has utility $1$ for each vertex agent, and $0$ for each edge agent. Call them {\em vertex items}.
\item For each edge $e$, there is an item $i(e)$. It has utility $1 - \alpha$ for $a(e)$, and $0$ for all other agents. Call them {\em edge items}.
\item For each vertex-edge pair $(v, e)$ with $v \in e$, there is an item $i(v, e)$. 
It has utility $\frac{1}{3}$ for $a(v)$ and $\alpha$ for $a(e)$, and $0$ for all other agents. 
Call them {\em shared items}. 
\end{itemize}

It is easy to verify that $n = N + M = 2.5N$ and $m = c_{min}N + M + 3N = (4.5 + c_{min})N$.

\paragraph{Completeness.}
Suppose $G$ has a vertex cover $C \subseteq V$ of size $c_{min}N$. We allocate items as follows.
\begin{itemize}
\item For each $v \in C$, $a(v)$ gets one vertex item. Her utility is 1.
\item For each $v \notin C$, $a(v)$ gets all three shared items that have nonzero utility. Her utility is 1.
\item Since $C$ is a vertex cover, for each edge $e$, $a(e)$ has either 1 or 2 shared items that have nonzero utility for her and are not taken by any vertex agent. $a(e)$ gets all these items and $i(e)$. 
\begin{itemize}
\item If she gets 1 shared item, her utility is 1. If she gets 2 shared items, her utility is $1 + \alpha$. 
\item Out of $3N$ shared items, $3(1 - c_{min})N$ items are taken by vertex agents and $3c_{min} N$ items are taken by edge agents. Therefore, $3(1 - c_{min}) N$ edge agents have 1 shared item, and $3(c_{min} - 0.5)N$ edge agents have 2 shared items. Let $\beta := 3(c_{min} - 0.5) \approx 0.03$. 
\end{itemize}
\end{itemize}
Therefore, the Nash social welfare of the above allocation is $((1 + \alpha)^{\beta N})^{1/n} = (1 + \alpha)^{\beta / 2.5}$.

\paragraph{Soundness.}
For soundness, we prove that if every vertex cover of $G$ has size at least $c_{max}N$, 
the optimal Nash social welfare of the instance produced by our reduction is low. 
First, we observe that in any optimal allocation, 
each vertex item should be assigned to some vertex agent, and each edge item $i(e)$ should be assigned to the corresponding agent $a(e)$. 
We prove the following lemmas that guarantee the existence of an optimal allocation of a certain form.
\begin{lem}
There exists an optimal allocation where each vertex agent has at most 1 vertex item.
\label{lem:opt1}
\end{lem}
\begin{proof}
Take an arbitrary optimal allocation $\cala$. If no vertex agent has more than 1 vertex item, the lemma is proved. 
Let $a(v)$ be an agent that has $x$ vertex items for $x \geq 2$. 
Note that $u_{a(v)}(\cala) \geq x \geq 2$.
Since the number of vertex items $c_{min}N$ is less than the number of vertex agents $N$, there is another vertex agent $a(v\rq{})$ that has no vertex item. The utility of $a(v\rq{})$ only comes from shared items, so her utility $u_{a(v\rq{})}(\cala)$ is at most 1.
Consider another allocation $\cala\rq{}$ here $a(v)$ gives one of her vertex item to $a(v\rq{})$. 
Utilities of agents except $a(v)$ and $a(v\rq{})$ do not change, and the product of the utilities of $a(v)$ and $a(v\rq{})$ becomes
\[
u_{a(v)}(\cala\rq{}) 
u_{a(v\rq{})}(\cala\rq{}) = 
(u_{a(v)}(\cala) - 1)
(u_{a(v\rq{})}(\cala) + 1) 
\geq 
u_{a(v)}(\cala) u_{a(v\rq{})}(\cala),
\]
where equality holds only if $u_{a(v)}(\cala) = 2, u_{a(v\rq{})}(\cala) = 1$. Therefore, this new allocation never decreases the Nash social welfare. 
The lemma is proved by repeated applying this observation. 
\end{proof}
\begin{lem}
Let $\alpha \in (\frac{1}{3}, \frac{1}{2})$. 
There exists an optimal allocation where each vertex agent has at most 1 vertex item, and 
each shared item $i(v, e)$ satisfies the following rule. 
\begin{enumerate}
\item If $a(v)$ has 1 vertex item, $a(e)$ gets $i(v, e)$.
\item Elseif $a(e)$ has 1 shared item other than $i(v, e)$, $a(v)$ gets $i(v, e)$. 
\item Elseif $a(v)$ has 2 shared items other than $i(v, e)$, $a(e)$ gets $i(v, e)$. 
\item Else $a(v)$ gets $i(v, e)$.
\end{enumerate}
\label{lem:opt2}
\end{lem}
\begin{proof}
By Lemma~\ref{lem:opt1}, 
fix an optimal allocation $\cala$ where each vertex agent has at most 1 vertex item. 
Fix a shared item $i(v, e)$. 
For each step of the above rule, we assume that it is violated
and derive contradiction to the optimality of the current allocation 
by considering another allocation.

\begin{enumerate}
\item Assume towards contradiction that $a(v)$ has 1 vertex item but $a(v)$ also gets $i(v, e)$. 
It implies that $u_{a(v)}(\cala) \geq \frac{4}{3}$ and $u_{a(e)}(\cala) \leq 1$. 
Consider another allocation $\cala\rq{}$ where $a(v)$ gives $i(v, e)$ to $a(e)$. Utilities of other agents do not change and 
\[
u_{a(v)}(\cala\rq{}) 
u_{a(e)}(\cala\rq{}) 
= 
(u_{a(v)}(\cala) - 1/3)
(u_{a(e)}(\cala) + \alpha)
>
u_{a(v)}(\cala)
u_{a(e)}(\cala),
\]
since 
\[
\frac{(u_{a(v)}(\cala) - 1/3)}{u_{a(v)}(\cala)} \geq \frac{3}{4}, \qquad 
\frac{(u_{a(e)}(\cala) + \alpha)}{u_{a(e)}(\cala)} \geq 1 + \alpha, \qquad
\frac{3}{4} (1 + \alpha) > 1.
\]
This contradicts the optimality of $\cala$. 

\item Assume towards contradiction that $a(v)$ has no vertex item, $a(e)$ has 1 shared item other than $i(v, e)$, but $a(e)$ also gets $i(v, e)$. 
It implies that $u_{a(v)}(\cala) \leq \frac{2}{3}$ and $u_{a(e)}(\cala) = 1 + \alpha$. 
Consider another allocation $\cala\rq{}$ where $a(e)$ gives $i(v, e)$ to $a(v)$. Utilities of other agents do not change and 
\[
u_{a(v)}(\cala\rq{}) 
u_{a(e)}(\cala\rq{}) 
= 
(u_{a(v)}(\cala) + 1/3)
(u_{a(e)}(\cala) - \alpha)
>
u_{a(v)}(\cala)
u_{a(e)}(\cala),
\]
since 
\[
\frac{(u_{a(v)}(\cala) + 1/3)}{u_{a(v)}(\cala)} \geq \frac{3}{2}, \qquad 
\frac{(u_{a(e)}(\cala) - \alpha)}{u_{a(e)}(\cala)} = \frac{1}{1 + \alpha}, \qquad
\frac{3}{2} \cdot \frac{1}{1 + \alpha} > 1. 
\]
This contradicts the optimality of $\cala$. 

\item Assume towards contradiction that $a(v)$ has no vertex item but has all 3 shared items including $i(v, e)$, $a(e)$ has no shared item. 
It implies that $u_{a(v)}(\cala) = 1$ and $u_{a(e)}(\cala) = 1 - \alpha$. 
Consider another allocation $\cala\rq{}$ where $a(v)$ gives $i(v, e)$ to $a(e)$. Utilities of other agents do not change and 
\[
u_{a(v)}(\cala\rq{}) 
u_{a(e)}(\cala\rq{}) 
= 
(u_{a(v)}(\cala) - 1/3)
(u_{a(e)}(\cala) + \alpha)
>
u_{a(v)}(\cala)
u_{a(e)}(\cala),
\]
since 
\[
\frac{(u_{a(v)}(\cala) - 1/3)}{u_{a(v)}(\cala)} = \frac{2}{3}, \qquad 
\frac{(u_{a(e)}(\cala) + \alpha)}{u_{a(e)}(\cala)} = \frac{1}{1 - \alpha}, \qquad
\frac{2}{3} \cdot \frac{1}{1 - \alpha} > 1. 
\]
This contradicts the optimality of $\cala$. 

\item Assume towards contradiction that $a(v)$ has no vertex item and at most 1 shared item other than $i(v, e)$, 
but $a(e)$ gets $i(v, e)$. 
It implies that $u_{a(v)}(\cala) \leq \frac{1}{3}$ and $u_{a(e)}(\cala) \geq 1$. 
Consider another allocation $\cala\rq{}$ where $a(e)$ gives $i(v, e)$ to $a(v)$. Utilities of other agents do not change and 
\[
u_{a(v)}(\cala\rq{}) 
u_{a(e)}(\cala\rq{}) 
= 
(u_{a(v)}(\cala) + 1/3)
(u_{a(e)}(\cala) - \alpha)
>
u_{a(v)}(\cala)
u_{a(e)}(\cala),
\]
since 
\[
\frac{(u_{a(v)}(\cala) + 1/3)}{u_{a(v)}(\cala)} \geq 2, \qquad 
\frac{(u_{a(e)}(\cala) - \alpha)}{u_{a(e)}(\cala)} \geq 1 - \alpha, \qquad
2 (1 - \alpha) > 1. 
\]
This contradicts the optimality of $\cala$. 
\end{enumerate}
\end{proof}

Now we prove the main lemma for soundness.
Note that for any $\gamma > 0$, $(\frac{2(1 + \alpha)}{3})^{\gamma / 2.5} < 1$. 
\begin{lem}
For some universal constant $\gamma > 0$, 
if every vertex cover of $G$ has at least $c_{max} N$ vertices, 
the optimal Nash social welfare is at most $(\frac{2(1 + \alpha)}{3})^{\gamma / 2.5} (1 + \alpha)^{\beta / 2.5}$.
\end{lem}
\begin{proof}
Let $\cala$ be an optimal allocation that satisfies the conditions in Lemma~\ref{lem:opt2}.
Let $C \subseteq V$ be a set of vertices $v$ such that $a(v)$ gets 1 vertex item ($|C| = c_{min} N$), 
and $I := V \setminus C$. 
If the subgraph induced by $I$ has less than $(c_{max} - c_{min})N$ edges, 
there exists a vertex cover of size less than $c_{max}N$ vertices, since adding one arbitrary endpoint of each such edge to $C$ yields a vertex cover. 
This contradicts that every vertex cover of $G$ has size at least $c_{max}N$, so the subgraph induced by $I$ has at least 
$(c_{max} - c_{min})N$ edges.

By the first rule of Lemma~\ref{lem:opt2}, if some vertex agent $a(v)$ gets 1 vertex item, it does not get any shared item. 
By the fourth rule, if $a(v)$ does not get any vertex item, it has to get at least 2 shared items. 
Therefore, for each vertex $v \in V$, the possibility utility $u_{a(v)}(\cala) \in \{ \frac{2}{3}, 1 \}$. 
Let $I_3 \subseteq I$ be a set of vertices $v$ such that the corresponding agent $a(v)$ has no vertex item but all 3 shared items. 
Let $I_2 := I \setminus I_3$ be a set of vertices $v$ such that the corresponding agent $a(v)$ has no vertex item and 2 shared items. 
It is clear that $u_{a(v)}(\cala) = 1$ if $v \in C \cup I_3$ and $u_{a(v)}(\cala) = \frac{2}{3}$ if $v \in I_2$. 
For $i = 0, 1, 2$, let $E_i$ be the set of edges $e$ such that the corresponding agent $a(e)$ gets $i$ shared items. 
We deduce the following facts.

\begin{enumerate}
\item By definition, each edge $e$ of $E_2$ cannot have any endpoint in $I_3$. 
By the second rule of Lemma~\ref{lem:opt2}, it cannot have any endpoint in $I_2$. 
Therefore, it has to have both endpoints in $C$.
\item By definition, each edge $e$ of $E_1$ cannot be entirely contained in $I_3$.
By the first rule of Lemma~\ref{lem:opt2}, it cannot be entirely contained in $C$. 
Let $E_{1C} \subseteq E_{1}$ be the set of edges covered by $C$, and 
$E_{1I} := E_1 \setminus E_{1C}$. 
\item By the third rule of Lemma~\ref{lem:opt2}, all edges in $E_0$ have both endpoints in $I_2$. 
\end{enumerate}
Therefore, the number of edges induced by $I$ is $|E_{1I}| + |E_0| \geq (c_{max} - c_{min})N$. 
Since edges in $E_0$ have both endpoints in $I_2$ and 
edges in $E_{1I}$ have at least one endpoint in $I_2$, 3-regularity of $G$ implies 
$3|I_2| \geq |E_{1I}| + 2|E_0|$, so we can conclude that $|I_2| \geq \gamma  N$,
where $\gamma := \frac{(c_{max} - c_{min})}{3} \approx 0.0017$.

Furthermore, by the number of shared items, the number of vertices, the number of edges, we have the following three identities.
\begin{equation}
\label{eq1}
3 |I_3| + 2 |I_2| + 2|E_2| + |E_1| = 3N.
\end{equation}
\begin{equation}
\label{eq2}
|I_3| + |I_2| = (1 - c_{min}) N. 
\end{equation}
\begin{equation}
\label{eq3}
|E_2| + |E_1| + |E_0| = M = 1.5N. 
\end{equation}
There are five variables and three identities. Eliminating $|I_3|$ and $|E_1|$ gives 
\begin{align*}
\eqref{eq1} - 3 \eqref{eq2} - \eqref{eq3} :& -|I_2| + |E_2| - |E_0| = (1.5 - 3(1 - c_{min}))N \\
\Rightarrow & |E_2| = \beta N + |I_2| + |E_0|,
\end{align*}
where $\beta$ was defined to be $3(c_{min} - 0.5) \approx 0.03$. 

The Nash social welfare of $\cala$ is 
\begin{align*}
& \bigg( \big(\frac{2}{3}\big)^{|I_2|} (1 + \alpha)^{|E_2|} (1 - \alpha)^{|E_0|} \bigg)^{1/n} \\
=& \bigg( \big(\frac{2}{3}\big)^{|I_2|} (1 + \alpha)^{|I_2| + |E_0| + \beta N} (1 - \alpha)^{|E_0|} \bigg)^{1/n} \\
=& \bigg( \big(\frac{2(1 + \alpha)}{3}\big)^{|I_2|} (1 + \alpha)^{\beta N} \big((1 + \alpha)(1 - \alpha)\big)^{|E_0|} \bigg)^{1/n} \\
\leq& \bigg( \big(\frac{2(1 + \alpha)}{3}\big)^{\gamma N} (1 + \alpha)^{\beta N}  \bigg)^{1/n} \\
=& \big(\frac{2(1 + \alpha)}{3}\big)^{\gamma / 2.5} (1 + \alpha)^{\beta / 2.5} .
\end{align*}
\end{proof}
By our completeness and soundness properties,  if $G$ has a vertex cover of size at most $c_{min}N$, the reduced instance has the Nash social welfare at least $(1 + \alpha)^{\beta/2.5}$, while if every vertex cover of $G$ has size at least $c_{max}N$, the reduced instance has the Nash social welfare at most $(\frac{2(1 + \alpha)}{3})^{\gamma/2.5} (1 + \alpha)^{\beta/2.5}$.
Therefore, it is NP-hard to approximate the Nash social welfare within factor 
\[
\mu := (\frac{2(1 + \alpha)}{3})^{-\gamma/2.5} > 1.
\]
With $\alpha = \frac{1}{3}$ and $\gamma \approx 0.0017$, $\mu \approx 1.00008$. 
It still remains an open problem to close the gap between $1.00008$ and $2.889$. 

\bibliographystyle{abbrv}
\bibliography{../mybib}

\end{document}